\documentclass[twocolumn]{autart}  
\usepackage{graphicx}    
\usepackage{cite}
\usepackage{comment}
\usepackage{xcolor}

\usepackage{amsmath, amsfonts}
\newcommand{\reals}{\mathbb{R}}
\newcommand{\E}{\mathbb{E}}
\newcommand{\xdim}{d_x}
\newcommand{\bdim}{d_b}

\newcommand{\N}{\mathbb{N}}
\newcommand{\Dt}{\tau}

\newcommand{\PR}{\mathbb{P}}
\newcommand{\tf}{\psi}
\newcommand{\given}{\,|\,}
\newcommand{\dx}{\mathrm{d}x}
\newcommand{\V}{}
\newcommand{\wct}{\tau} 

\newcommand{\scn}{N_s}
\newcommand{\wcn}{N_c}

\newtheorem{assumption}{Assumption}

\usepackage{booktabs}

\newcommand{\DW}[1]{{\color{blue}[DW: #1]}}

\usepackage{algorithm}
\usepackage{algpseudocode}

\newtheorem{definition}{Definition}

\makeatletter
\def\thm@space@setup{%
  \thm@preskip=8pt plus 2pt minus 4pt
  \thm@postskip=\the\thm@preskip
}
\makeatother

\makeatletter
\g@addto@macro\normalsize{%
  \setlength\abovedisplayskip{6pt plus 2pt minus 4pt}
  \setlength\belowdisplayskip{6pt plus 2pt minus 4pt}
  \setlength\abovedisplayshortskip{0pt plus 3pt}
  \setlength\belowdisplayshortskip{4pt plus 2pt minus 2pt}
}
\makeatother

\begin{document}

\begin{frontmatter}

\title{Wiener-Space Cubature for Nonlinear Gaussian Filtering}

\author[Basis,Cornell]{Luke Snow}\ead{las474@cornell.edu}, \author[Basis,MIT]{Daniel Waxman},
\author[Basis]{Matthew E. Levine}

\address[Basis]{Basis Research Institute, Cambridge, MA}
\address[Cornell]{Cornell University, Ithaca, NY}  
\address[MIT]{Massachusetts Institute of Technology, Cambridge, MA}



\begin{abstract} Continuous-discrete filtering is a classical problem in signal processing and control theory, whereby the state driving a stochastic differential equation (SDE) must be inferred from noisy, partial discrete-time observations. Corresponding filtering algorithms are often formed as discrete-time methods and applied to a discretization of the SDE --- an approach limited by the incurred discretization error --- or model a type of Gaussian moment ordinary differential equation (ODE) flow --- which is approximate for nonlinear dynamics. We introduce a new method, the Wiener-space cubature Kalman filter (WSC-KF), for Gaussian filtering in this continuous-discrete setting. The WSC-KF combines spatial sigma-point representations with Wiener space cubature integration, effectively approximating iterated integral terms of the driving Brownian motion to higher-order than standard discretization (e.g., It\^o-Taylor expansions), or approximate moment-ODE flows (e.g., as in the unscented Kalman filter). Theoretically, we characterize the path-law approximation degree of WSC-KF through a tunable model order parameter. We prove that, depending on the model order chosen, WSC-KF provides weak order bounds comparable to or dominating that of a cubature Kalman filter using a first-order It\^o predict step. Numerically, we confirm that rate on a globally Lipschitz L\'evy-area diffusion.  Degree-$5$ Wiener cubature matches mixed iterated integrals that existing continuous-discrete Kalman predictors omit, while remaining efficiently computable; this is the mechanism enabling higher-order weak path-law approximation.
\end{abstract}

\end{frontmatter}

\section{Introduction}
\label{sec:intro}
Nonlinear continuous-discrete filtering aims to estimate the underlying state of a nonlinear stochastic dynamical system, modeled as a stochastic differential equation (SDE) from a noisy, discrete sequence of (potentially partial) observations. Standard approaches exploit sequential Bayesian inference to compute posterior estimates of the state at each time, given the observations and assumed dynamics. In the special case where the SDE and observation map are linear with Gaussian noise, the Kalman filter \cite{kalman1960new} provides an exact, recursive formula for this Gaussian posterior. In the case when the model dynamics or observations are truly nonlinear, one must generally resort to approximate solutions. In this work we focus on the latter regime. 

Approximate filters may be loosely divided into two categories \cite{arasaratnam2010cubature}:
so-called ``type-I'' filters compute the state evolution by directly approximating the solution of the Fokker--Planck--Kolmogorov (FPK) equation, e.g., by numerical solvers \cite{beard1999nonlinear,daum2006meshfree} or Monte Carlo \cite{berzuini1997dynamic,raviart2006analysis}. This class includes sequential Monte Carlo (particle) filters \cite{doucet2001introduction} and ensemble Kalman filters (CD-EnKF) \cite{masaya2021ensemble}.
Type-II filters instead compute a finite number of summary statistics of the posterior, typically under a Gaussian closure. This class includes the continuous-discrete Extended Kalman Filter (CD-EKF) \cite{ribeiro2004kalman}, continuous-discrete Unscented Kalman Filter (CD-UKF) \cite{wan2000unscented},
and cubature Kalman filter (CD-CKF), the
series-expansion UKF (SE-UKF) \cite{lyons2014series}, and the Taylor moment expansion (TME) filter \cite{zhao2020taylor} . Type-I filters are more expressive, but inherit a more difficult approximation problem, whereas type-II filters trade off expressivity for more scalable computation. The main challenge in this regime is that nonlinearities result in non-Gaussian flows, which are difficult to approximate with conventional SDE discretizations or Gaussian moment-ODEs common to type-II filters. Furthermore, methods which can accurately approximate these non-Gaussian flows, such as sequential Monte Carlo methods, tend to be numerically brittle and also prohibitively expensive even in moderate dimensions. 

Previous work has acknowledged that the approximation of SDE transitions via discretization or Gaussian moment ODEs may be insufficient. The SE-UKF \cite{lyons2014series}, for example, instead replaces Brownian motion via a truncated Fourier series, and the TME \cite{zhao2020taylor} expands the first two moment maps in time through the infinitesimal generator. These approaches, however, do not match the mixed iterated integrals of the driving Brownian motion, which become manifest through L\'evy area terms arising from multiplicative noise. 



We instead draw on tools from rough path theory \cite{friz2014course}, specifically the theory of cubature formulae on Wiener space \cite{lyons2004cubature}, to provide a filter which \emph{efficiently} models non-Gaussian flow through \emph{deterministic} evaluations. Specifically, each filtering step solves a system of ODEs along \emph{Wiener space cubature} paths which can be shown to capture high-order statistics of the driving Brownian motion \cite{litterer2026high,lyons2004cubature}; this is juxtaposed to, for example, It\^o-Taylor expansions beyond weak order $1.0$ (Euler--Maruyama), which require solving (generally intractable) iterated stochastic integrals \cite{kloeden1992numerical}. Wiener space cubature paths are constructed to match these iterated integral terms, thus enabling higher-order accuracy whilst gaining the computational benefits of parallel ODE solvers. These constructions were originally designed for computing efficient and stable solutions of parabolic PDEs, but have recently seen broader application, such as neural SDE optimization \cite{snow2025efficient}. 



Algorithmically, the WSC-KF is obtained as a drop-in modification to any Gaussian filter based on a sigma-point approximation (e.g., the CD-UKF or CD-CKF): instead of developing a spatial quadrature through a first-order It\^o step of the dynamics, each cubature node is pushed forward by integrating the same controlled ODE along a finite set of deterministic Wiener-space cubature paths. The point cloud is then collapsed back to a Gaussian approximation at each observation iterate. Similar Wiener-cubature prediction mechanisms have been exploited in part for particle filter (type-I) architectures \cite{crisan2013kusuoka,lee2013adaptive}. The advantage of WSC-KF stems from designing these in a type-II architecture.

The main advantages of the WSC-KF are:
\begin{enumerate}
\item \textit{Signature-matching path law.} The predict step exploits a high-order Wiener cubature approximation of the SDE path law. Degree $m \geq 5$  paths match expected mixed iterated integrals which are not captured by moment-ODE Kalman predictors, the SE-UKF Fourier truncation, or the TME closure of the first two moments.
\item \textit{Tunable weak order relative to CD-CKF.} We characterize the weak error in the predict step of the WSC-KF in terms of the chosen integration ``order'' for the Wiener space cubature rule (Theorem~\ref{thm:WSC-KF-error}). We further compare to the CD-CKF (Corollary~\ref{cor:cdckf-comparison}), which is algorithmically similar to our method but uses an It\^o--Taylor expansion, showing that the rates coincide at order $m = 3$, and are better for WSC-KF at $m = 5$. 
\item \textit{Deterministic ODE forecast.} The Wiener-cubature paths are precomputed and deterministic. Thus, there is no Brownian sampling in the predict step, in contrast to the CD-EnKF and particle-based filters. This enables several computationally desirable properties, such as improved time- and memory-efficiency for gradient computation.
\end{enumerate}

Broadly, this paper motivates how Wiener-space cubature can be inserted into the prediction step of a continuous-discrete Gaussian filter, to yield a tunable high-order weak approximation of the transition law, while retaining Gaussian-update structure. This paper is organized as follows. Section~\ref{sec:cdfilt} introduces the continuous-discrete nonlinear filtering problem, Section~\ref{sec:wsc} introduces the mathematical preliminaries of Wiener space cubature formulae, Section~\ref{sec:wsc-kf} constructs and analyzes our WSC-KF filter, and Section~\ref{sec:num} illustrates the practical efficacy of the WSC-KF in a numerical implementation.\footnote{Full open-source software for this filter, as well as several thorough experimental demonstrations, are provided at \texttt{https://github.com/LukeSnow0/roughpath-filtering}.}

\section{Continuous-Discrete Nonlinear Filtering}
\label{sec:cdfilt}
Here we introduce the continuous-discrete nonlinear filtering framework, and motivate our contribution in relation to existing state-of-the-art filtering algorithms.

\subsection{System Dynamics Model}
Continuous-discrete filtering \cite{jazwinski2007stochastic} refers to the problem of state estimation in a system governed by a continuous stochastic dynamics model over the \emph{state} $x_t \in \reals^{d_x}$, with observations $y_{t_k} \in \reals^{d_y}$ at a finite set of times $t_{1}, \dots, t_{K}$:
\begin{align}
\label{eq:model}
\mathrm{d}x_t &= f(t, x_t) \, \mathrm{d}t + G(t,x_t)\, \mathrm{d} W_t, \,\,\, x_0 \sim \mathcal{N}(\bar{x}_0, P_0) \\
y_{k} &\sim p(y_{k} \given t_k, x_{t_k}).
\label{eq:obs}
\end{align}
Here \eqref{eq:model} is an It\^o stochastic differential equation with $\V{f}: [0,T] \times \reals^{\xdim} \to \reals^{\xdim}$, $\V{G}: [0,T] \times \reals^{\xdim} \to \reals^{\xdim} \times \reals^{d_b}$ as nonlinear drift and diffusion functions, respectively, and where $\V{W}_t$ is a $\bdim$-dimensional standard Brownian motion and $\mathcal{N}(\bar{x}_0, P_0)$ is the distribution of initial condition $x_0$. Notably, \eqref{eq:obs} posits a generic discrete observation model, with additive Gaussian noise as a special case.

\subsection{Filtering Equations}
The goal of nonlinear Bayesian filtering is to recursively form a posterior estimate of the state $\V{x}_{t_k}$ from the observations $\{\V{y}_{k}\}_{k=1}^K$ and the model \eqref{eq:model},\,\eqref{eq:obs} via sequential Bayesian inversion. Specifically, denote
\begin{equation}
\label{eq:post}
\pi_k(\dx) = \PR(\V{x}_{t_k} \in \dx \given \V{y}_{1:k})
\end{equation}
as the filtered Bayesian posterior distribution at time $t_k$. Updating $\pi_k$ to $\pi_{k+1}$ at the next iteration consists of a \emph{prediction} step and an \emph{update} step.

\textbf{Prediction Step.} In the prediction step at iteration $k$, the previous filtering distribution $\pi_{k-1}$ is propagated through the transition kernel induced by the model \eqref{eq:model} to form the predicted distribution $\pi_{k|k-1}$, given by
\begin{equation}
\label{eq:filt_predict}
\pi_{k|k-1}(\dx) = \int \PR(\V{x}_{t_k}\in \dx \given \V{x}_{t_{k-1}} = \V{x}')\, \pi_{k-1}(\dx') \tag{predict step} 
\end{equation}

Under standard assumptions (namely, let $\pi_{k \given k-1}$ admit a density w.r.t. Lebesgue, and standard regularity and nondegeneracy assumptions hold; cf. \cite{stroock2007multidimensional}), it is well known that \eqref{eq:filt_predict} corresponds at time $t_k$ to the density
$\pi(t,x) := \PR(\V{x}_t = x | \V{y}_{1:k})$ which evolves via the FPK equation, the exact solution of which is intractable for general SDEs \cite{fuller1969analysis,jazwinski2007stochastic}.

\textbf{Update Step.} In the update step, the new observation $\V{y}_k$ is incorporated through the Bayesian inversion
\begin{equation}
\label{eq:filt_update}
\pi_k(\dx_{t_k}) \propto \PR(\V{y}_k|\V{x}_{t_k})\pi_{k|k-1}(\dx_{t_{k}}) \tag{update step}
\end{equation}
which is scaled to a density by the marginal likelihood $\PR(\V{y}_k | \V{y}_{1:k-1})$. Here $\PR(\V{y}_k|\V{x}_{t_k})$ denotes the conditional probability measure induced by the observation model \eqref{eq:obs}. Observe that these recursive filtering equations operate in discrete-time in accordance with the observation rate, but utilize the continuous-time model \eqref{eq:model} in computing \eqref{eq:filt_predict}.\footnote{For continuous-time observation processes, filtered states are given by the Kushner-Stratonovich or Zakai equations; though some efficient algorithms are known \cite[Sec. 10.3--10.4]{sarkka2019applied}, we focus only on the CD-SSM problem in this work.}

\textbf{Approaches for Prediction.} In this work we innovate along the \emph{prediction} step, and we thus outline here the existing approaches. Existing approaches to \eqref{eq:filt_predict} can be delineated along the type-I/type-II line outlined above. Type-I methods approximate the FPK equation directly, via PDE solvers \cite{beard1999nonlinear} or particle ensembles \cite{doucet2001introduction, raviart2006analysis, masaya2021ensemble}, capturing arbitrary non-Gaussian structure at the cost of poor dimensionality scaling and particle degeneracy issues. Type-II methods instead approximate the SDE's one-step transition under a Gaussian closure, via discretization (e.g., It\^o-Taylor \cite{kloeden1992numerical}) or direct integration of Gaussian moment ODEs, yielding filters such as CD-EKF \cite{jazwinski2007stochastic}, CD-UKF/CD-CKF \cite{wan2000unscented, arasaratnam2010cubature}, TME \cite{zhao2020taylor} or SE-UKF \cite{lyons2014series}. Our approach retains the type-II Gaussian-closure structure for its stability and efficiency, but replaces the It\^o--Taylor discretization step with evaluations of Wiener-space cubature ODEs, which deterministically approximate the transition law itself to tunable high order, as we motivate next.

\section{Gaussian and Wiener Space Cubature} \label{sec:wsc}

The WSC-KF predict step is a product cubature rule. A spatial cubature rule of degree $r$ replaces the current Gaussian by $\scn$ weighted nodes in $\reals^{\xdim}$. Independently, a Wiener-space cubature rule of degree $m$ replaces the law of Brownian motion on an interval of length $\tau$ (e.g., $t_k - t_{k-1}$) by $\wcn$ deterministic paths. Each of the $(\scn,\wcn)$ pairs is integrated as an ODE of length $\tau$; the cloud is then collapsed to a Gaussian. Section~\ref{sec:spat_cub} recalls the spatial rule (equivalent to CD-CKF). Section~\ref{sec:ws_cub} constructs the path rule and states the weak error $O(\tau^{(m+1)/2})$ for a fixed initial condition $x_0$. In Section~\ref{sec:wsc-kf} these two cubature rules are combined into our filter architecture.

\subsection{Spatial Cubature}
\label{sec:spat_cub}
Our WSC-KF lies in the class of Gaussian approximation (type-II) filters. Accordingly, at each assimilation step we approximate the filtering distribution by a Gaussian 
$\hat{\pi}_k = \mathcal{N}(\bar{x}_k, P_k)$ and represent this Gaussian by a finite collection of deterministic sigma points (spatial cubature nodes); this sigma-point representation is motivated in the UKF \cite{wan2000unscented} and CD-CKF \cite{arasaratnam2010cubature}. Specifically, a \emph{spatial}
cubature rule of degree $r$ consists of nodes $\{\zeta_i\}_{i=1}^{\scn}\subset\reals^{\xdim}$ and weights $\{c_i\}_{i=1}^{\scn}$ satisfying 
\[\int_{\reals^{\xdim}} g(x) \,\mathrm{d}\hat{\pi}_k(x) = \sum_{i=1}^{\scn} c_i g(\zeta_i)\]
for every polynomial $g$ of degree at most $r$. Consequently, 
$\E_{\hat{\pi}_k}[\psi(x)] \approx \sum_{i=1}^{\scn} c_i \psi(\zeta_i)$ for sufficiently smooth test functions $\psi$. Many sigma-point approximations exist (and WSC-KF remains generally compatible with them all); throughout our numerics, we use the popular method of \cite{arasaratnam2010cubature}, which uses $\scn = 2\xdim$ equally weighted sigma points $\zeta_i^{\pm} = \bar{x}_k \pm \sqrt{\xdim} \hat{U}_k e_i, \quad c_i = \frac{1}{2\xdim}$
where $\hat{U}_k\hat{U}_k^\top = P_k$ and $e_i$ is the $i$'th unit basis vector in $\reals^{\xdim}$. Fifth-degree spatial cubature nodes are constructed in the high-order cubature Kalman filter \cite{jia2013high}; this spatial representation is also compatible with our algorithm. Our contribution is distinguished from \cite{arasaratnam2010cubature, jia2013high} entirely in \emph{how these spatial nodes are propagated forward through the model dynamics}, approximating the path-law of the forward dynamics \eqref{eq:model}\footnote{This will allow us to precisely isolate our contribution as an efficient and high-order procedure for this path-law approximation, improving on \cite{arasaratnam2010cubature} directly and other filters indirectly by inheriting the benefits of \cite{arasaratnam2010cubature}.}. Specifically, we employ \emph{cubature on Wiener space} to approximate this path-law. 

\subsection{Wiener Space Cubature}
\label{sec:ws_cub}
In this section we present preliminary details on Wiener space cubature paths and weights, which serve as high-order deterministic approximations for SDE solution moments. 


Throughout this subsection the initial condition $x_0\in \reals^{\xdim}$ is fixed as a Dirac delta distribution. Averaging over non-Dirac $p_0$ is achieved in Section~\ref{sec:wsc-kf}, where our full filter architecture is outlined. We write \eqref{eq:model} in Stratonovich form. Following \cite{kloeden1992numerical}, adjoin the drift as a zeroth vector field, so that there are $\bdim+1$ fields and $\bdim$ Brownian components: 
\begin{equation}
\mathrm{d}x_t = V_0(x_t)\,\mathrm{d}t + \sum_{i=1}^{\bdim} V_i(x_t) \circ \mathrm{d}W^i_t, \qquad x_t\big|_{t=0}=x_0
\label{eq:strat_sde}
\end{equation}
with $V_i: \reals^{\xdim} \to \reals^{\xdim}$, $i=0,\ldots,\bdim$. The sum in the Taylor formula below runs over the same index set $0,\ldots,\bdim$, with the convention $\mathrm{d}W^0_t := \mathrm{d}t$. 
Time-dependent $(f,G)$ reduce to~\eqref{eq:strat_sde} by state augmentation $x\leftarrow(t,x)$; we work with the autonomous form without loss of generality.

\textbf{Cubature Path Definition.} 
Let $\mathcal A_m$ be the index set
$\mathcal A_m =\bigl\{(i_1,\ldots,i_\ell)\in\{0,\ldots,d_b\}^\ell: \ell+\mathrm{card}\{j:i_j=0\}\le m \bigr\}.$
For $\psi\in C^\infty_b(\reals^{d_x};\reals)$
(bounded, infinitely differentiable), the weak Taylor
formula~\cite{kloeden1992numerical} gives, at time $\tau>0$,
\begin{align}
\begin{split}
&\mathbb{E}[\psi(x_\tau)]= \sum_{(i_1,\ldots,i_\ell)\in\mathcal A_m} \left(V_{i_1:i_\ell}\psi\right)(x_0)\\&\,\,\,\cdot\mathbb E\Bigl[\int_{0<s_1<\cdots<s_\ell<\tau}\circ\mathrm{d}W^{i_1}_{s_1}\cdots\circ\mathrm{d}W^{i_\ell}_{s_\ell}\Bigr]
  + O\bigl(\tau^{(m+1)/2}\bigr). \raisetag{50pt}
\label{eq:sol_exp}
\end{split}
\end{align}
Here each $V_i$ acts on the function $\psi$ as a differential operator; specifically $(V_i\psi)(x) := \langle\nabla\psi(x), V_i(x)\rangle_{\reals^{\xdim}}$ and we write $V_{i_1:i_\ell}\psi := V_{i_1}(V_{i_2}(\cdots (V_{i_l}\psi)\cdots))$ as the function formed by iterating this differentiation along each indexed vector field.\footnote{The iterated Stratonovich integral is computed as: $\int_{0<s_1<s_2<\tau}\circ dW_{s_1}^{i_1} \circ dW_{s_2}^{i_2} = \int_{0 < s_2 < \tau} \left(\int_{0<s_1< s_2}\circ dW_{s_1}^{i_1}\right) \circ dW_{s_2}^{i_2}$, with the entire mixed integral extending this to iteration over all indices.}
The times \(s_1<\cdots<s_\ell\) parametrize the simplex of
an iterated integral of length \(\ell\le m\) (and, in particular, are not directly related to the observation times or the grid of our numerical integration).

\begin{definition}[Cubature Paths and Weights~\cite{lyons2004cubature}]
\label{def:cub}
Let $C^{0}_{\mathrm{bv}}([0,\tau];\mathbb R^{\bdim})$ denote
continuous paths in $\reals^{\bdim}$ of bounded variation on $[0,\tau]$. A collection of paths $w_1,\ldots,w_{N_c} \in C^{0}_{\mathrm{bv}}([0,\tau];\mathbb R^{\bdim}) $ in that space, together with weights $\lambda_1,\ldots,\lambda_{N_c}>0$ s.t. $\sum_{j=1}^{N_c}\lambda_j = 1$, defines a cubature formula of degree $m$ at time $\tau$ if, for every $(i_1,\ldots,i_\ell)\in\mathcal A_m$,
\begin{align}
\begin{split}
  &\mathbb E\Bigl[
    \int_{0<s_1<\cdots<s_\ell<\tau} \circ\mathrm{d}W^{i_1}_{s_1}\cdots\circ\mathrm{d}W^{i_\ell}_{s_\ell} \Bigr] \\&= \sum_{j=1}^{N_c} \lambda_j \int_{0<s_1<\cdots<s_\ell<\tau} \mathrm{d}w^{i_1}_j(s_1)\cdots\mathrm{d}w^{i_\ell}_j(s_\ell).
\label{eq:itr_eq}
\end{split}
\end{align}
\end{definition}
\vspace{-0.5cm}
In \cite{lyons2004cubature} it is proved that, for any $m\in\N$ and $\wct\geq 0$, there exist ${\wcn} \leq \text{card}(\mathcal{A}_m)$ cubature paths and weights\footnote{$\text{card}(\cdot)$ denotes cardinality.} defining a cubature formula on Wiener space of degree $m$ at time $\wct$. Furthermore, \emph{explicit} constructions of cubature paths and weights of degree $3, 5$ and $7$, for any $\tau>0$, have been developed \cite{lyons2004cubature,litterer2026high}. The key insight is that \emph{we may gain a high-order approximation of   $\E[\tf(x_\wct)]$ by solving a weighted system of ordinary differential equations} with respect to cubature paths $\{w_j\}_{j=1}^{\wcn}$. Indeed, we capture this insight in Proposition~\ref{prop:cub_ode}. First, we must specify a standard structural assumption on the vector fields $\{\V{V}_i\}_{i=0}^{\bdim}$.

\begin{assumption}[Smoothness, Polynomial Growth]
\label{as:smooth}
For a fixed Wiener cubature degree $m$, the vector
fields $\V{V}_0,\ldots,\V{V}_{\bdim}$ in \eqref{eq:strat_sde} belong to
$C^{m+1}(\mathbb{R}^{\xdim};\mathbb{R}^{\xdim})$, and are globally Lipschitz continuous in x.
Moreover, for every multi-index $\alpha$ satisfying
$|\alpha|\le m+1$, there exist constants
$C_\alpha,\rho_\alpha>0$ such that $\|\partial^\alpha \V{V}_i(x)\|
\le
C_\alpha\bigl(1+\|x\|^{\rho_\alpha}\bigr),\,
i=0,\ldots,\bdim$.
\end{assumption}

\begin{prop}[Cubature ODE Approximation \cite{lyons2004cubature}]
\label{prop:cub_ode}
Let $\Phi_{\wct,x_0}(w)$ be the solution at time $\wct$ of the ODE
\begin{equation}
\label{eq:cub_ode}
\mathrm{d}\widetilde x_t =  \sum_{i=0}^{\bdim}\V{V}_i(\widetilde x_t)\mathrm{d}w^i(t), \quad \widetilde x_0 = x_0.
\end{equation}
with the convention $\mathrm{d}w^0(t):=\mathrm{d}t$.
Then, under Assumption~\ref{as:smooth}, we have, for any bounded and infinitely differentiable function $\psi \in C_b^{\infty}$, $\psi : \reals^{\xdim} \to \reals$,
\begin{equation}
\label{eq:cub_approx}
    \bigg|\E(\tf(x_{\wct})) - \sum_{j=1}^{\wcn} \lambda_j \tf(\Phi_{\wct,x_0}(w_j))\bigg| = O(\wct^{(m+1)/2})
\end{equation}
\end{prop}
\vspace{-0.3cm}
Proposition~\ref{prop:cub_ode} follows by expanding the ODE \eqref{eq:cub_ode} through its Taylor expansion, see \cite{lyons2004cubature}; indeed, the iterated integrals of cubature paths match the expected iterated Brownian integrals in \eqref{eq:itr_eq} by Definition~\ref{def:cub}. 


\textbf{Path Recombination.}
\label{sec:recom}
The error in~\eqref{eq:cub_approx} is for one local cubature formula of length $\tau$.  To reduce it one may partition
$[0,\tau]$ into $N_\tau$ subintervals $0=\tau_0<\cdots<\tau_{N_\tau}=\tau$ and concatenate the
same degree-$m$ formula on each piece \cite[Proposition 3.2]{lyons2004cubature}.  This grid is a forecast-interval partition; it is distinct from the observation times \(t_k\) and from the iterated-integral simplex \(s_1<\cdots<s_\ell\)
in~\eqref{eq:itr_eq}.  This concatenated collection comprises
\(N_c^{N_\tau}\) paths. Writing \(\delta_n:=\tau_n-\tau_{n-1}\), the resulting weak error from solving \eqref{eq:cub_ode} along this collection of concatenated paths is
\begin{equation}
\label{eq:recom_err}
O\Biggl(\delta_{N_\tau}^{1/2} + \sum_{n=1}^{N_\tau-1} \delta_n^{(m+1)/2}\,(\tau-\tau_n)^{-m/2}\Biggr).
\end{equation}
High-order recombination~\cite{litterer2012high} trims the number of these concatenated paths while preserving moment approximations, reducing the cardinality from exponential in \(N_\tau\) to polynomial; for brevity we refer to~\cite{snow2025efficient} for the construction and analysis. These recombination methods form a complementary and orthogonal set of numerical techniques that can be used in an implementation of WSC-KF. 

\section{Wiener-Space Cubature Kalman Filter}
\label{sec:wsc-kf}
Here we develop the Wiener-Space Cubature Kalman Filter (WSC-KF), and discuss its structural advantages and trade-offs relative to existing approaches. 

\subsection{Preliminaries}

The Stratonovich fields $V_0,\ldots,V_{\bdim}$ are those of
Section~\ref{sec:wsc}.  Observations $y_1,\dots, y_K$ arrive at $t_1,\ldots,t_K$ with forecast
length $\tau:=t_k-t_{k-1}$.\footnote{We assume this interval is fixed over $k\in\N$. However, our filter straightforwardly extends to non-uniform observation intervals by scaling the cubature paths appropriately (cf. \cite[Proposition 2.5]{lyons2004cubature})}
%
%
Recall $\Phi_{\tau,x_0}(w)$ is the solution at time $\wct$ of \eqref{eq:cub_ode}, with initial condition $x_0$ and cubature paths and weights $\{w_j, \lambda_j\}_{j=1}^{\wcn}$ of degree $m$ at time $\tau$,\footnote{Observe that $\{w_j, \lambda_j\}_{j=1}^{\wcn}$ are \textit{independent} from the Stratonovich vector fields $\{\V{V}_i\}_{i=0}^{\bdim}$, and can be pre-computed once in a pre-processing stage before filtering. For brevity, we leave the complexity of recombined construction of these paths to the exposition in \cite[Sec. 2--3]{snow2025efficient}.} that is
\begin{equation}
\label{eq:Phi_dyn}
\Phi_{\tau,x_0}(w) = x_0 + \int_0^{\wct}\sum_{i=0}^{\bdim}\V{V}_i\left(\Phi_{s,x_0}(w)\right)dw^i(s)
\end{equation}
The WSC-KF is a Gaussian approximation filter, that is, at each $k$ we form a Gaussian $\hat\pi_k = \mathcal{N}(\bar{\V{x}}_k, \V{P}_k)$, with mean $\bar{\V{x}}_k$ and covariance $P_k$, to approximate $\pi_k(\dx)$. Updating from $\hat\pi_{k-1}$ to $\hat\pi_k$ proceeds in the following steps:

\subsection{WSC-KF Algorithm}
\textbf{Spatial Cubature Approximation.} We form $\scn$ spatial cubature nodes $\{\zeta_i\}_{i=1}^{\scn}$, with weights $\{c_i\}_{i=1}^{\scn}$, approximating $\hat{\pi}_{k-1} = \mathcal{N}(\bar{x}_{k-1},P_{k-1})$. A standard choice, matching the CD-CKF \cite{arasaratnam2010cubature}, is to take $\scn=2\xdim$ third-degree spherical-radial cubature nodes, as in \cite{arasaratnam2010cubature} and recalling Section~\ref{sec:spat_cub}. Fifth-degree spatial cubature node constructions are provided in \cite{jia2013high}. 

\textbf{Predict Step (Wiener Cubature Forecasting)}. For each spatial node $\zeta_i$ and Wiener cubature path $w_j$, integrate the controlled Stratonovich ODE \eqref{eq:Phi_dyn} to evaluate \[x_{i,j} := \Phi_{\tau,\zeta_i}(w_j),\] assign the weight $\alpha_{i,j} := c_i\lambda_j$, and collapse this weighted cloud to approximate Gaussian moments:
\begin{align}
\begin{split}
\label{eq:fore_stat}
    \bar{x}_{k|k-1} &= \sum_{i=1}^{\scn}\sum_{j=1}^{\wcn} \alpha_{i,j} x_{i,j} \\
    P_{k|k-1} &= \sum_{i,j}\alpha_{i,j}(x_{i,j} - \bar{x}_{k|k-1})(x_{i,j} - \bar{x}_{k|k-1})^{\top} \raisetag{50pt}
\end{split}
\end{align}
\textbf{Update Step.} Assume the conditional observation moments of~\eqref{eq:obs} exist and
are finite,
\[\mu_k(x):=\mathbb E[y_k\mid x_{t_k}=x],\quad  S_k(x):=\mathrm{Cov}(y_k\mid x_{t_k}=x).
\]
Set $\bar z_{k|k-1} = \sum_{i,j}\alpha_{i,j}\,\mu_k(x_{i,j})$, $P_{zz} = \sum_{i,j}\alpha_{i,j} \Bigl(S_k(x_{i,j}) +(\mu_k(x_{i,j})-\bar z_{k|k-1}) (\mu_k(x_{i,j})-\bar z_{k|k-1})^\top \Bigr)$, $P_{xz} = \sum_{i,j}\alpha_{i,j}(x_{i,j}- \bar x_{k|k-1}) (\mu_k(x_{i,j})-\bar z_{k|k-1})^\top$.
The Kalman gain and corrected state are then
$K_k = P_{xz}P_{zz}^{-1},\,\, \bar x_k = \bar x_{k|k-1} + K_k(y_k-\bar z_{k|k-1}),\,\, P_k = P_{k|k-1}-K_k P_{zz}K_k^\top$.
Additive Gaussian \(y_k=h(x_{t_k})+\varepsilon_k\) with
\(\varepsilon_k\sim N(0,R_k)\) independent of \(x\) is the special case
\(\mu_k=h\), \(S_k\equiv R_k\).

\subsection{WSC-KF Analysis}

In this section we provide our main theorem, which relates the accuracy of the WSC-KF filter to that of the CD-CKF, which is its most structurally related predecessor. Throughout we let $\|\cdot\|$ denote $L_2$ norm, and $\|\cdot\|_F$ denote Frobenius norm.

\textbf{Specification 1}: 
Fix $k$ and let $\hat\pi_{k-1}=\mathcal N(\bar x_{k-1},P_{k-1})$ be the Gaussian
approximation entering the predict step. Let $\{\zeta_i\}_{i=1}^{\scn}, \{c_i\}_{i=1}^{\scn}$ be spatial cubature nodes and weights as in Section~\ref{sec:spat_cub}. Write $\mu(x):=\mathbb E[x_\tau\mid x_0=x],\,\Sigma(x):=\mathrm{Cov}(x_\tau\mid x_0=x)$ for the exact one-step conditional mean and covariance of \eqref{eq:model}. Define the \emph{exact one-step targets}, obtained by pushing $\hat\pi_{k-1}$ through the exact transition law,
$x_{k|k-1} := \int \mu(x)\,\hat\pi_{k-1}(dx), \Sigma_{k|k-1} := \int \Sigma(x)\,\hat\pi_{k-1}(dx) + \int \big(\mu(x)-x_{k|k-1}\big)^{\otimes 2}\,\hat\pi_{k-1}(dx),$
where $a^{\otimes 2}:=aa^\top$. Let $\tilde x_{k|k-1}$ and $\tilde\Sigma_{k|k-1}$ be
the same quantities with $\hat\pi_{k-1}$ replaced by its spatial-cubature point-mass representation, still evaluated against the exact transition law:
\begin{align*}
  \tilde x_{k|k-1} &:= \textstyle\sum_{i=1}^{\scn}c_i\mu(\zeta_i), \\
  \tilde\Sigma_{k|k-1} &:= \textstyle\sum_{i=1}^{\scn}c_i\Sigma(\zeta_i) + \textstyle\sum_{i=1}^{\scn}c_i\big(\mu(\zeta_i)-\tilde x_{k|k-1}\big)^{\otimes 2}.
\end{align*}
The spatial-cubature errors are then
\begin{equation*}
  \varepsilon_{\scn,x}:=\big\|\tilde x_{k|k-1}-x_{k|k-1}\big\|,\varepsilon_{\scn,\Sigma}:=\big\|\tilde\Sigma_{k|k-1}-\Sigma_{k|k-1}\big\|_F.
\end{equation*}
They are incurred purely by the spatial rule and depend only on $\scn$ and
$\hat\pi_{k-1}$, not on how $\mu,\Sigma$ are subsequently approximated. In
particular they are identical for WSC-KF and for the CD-CKF of \cite{arasaratnam2010cubature},
since both use the same spatial rule. 

The following is our main result, providing higher-order asymptotic approximation guarantees over the CD-CKF in the predict step. 



\begin{thm}[WSC-KF Error Decomposition]
\label{thm:WSC-KF-error}
Under Assumption~\ref{as:smooth}, the notation of Specification 1 and the WSC-KF Algorithm with Wiener-cubature degree $m$ and $\wct_0 > 0$, there exist constants $C_x, C_\Sigma$ independent of $\tau \in (0,\tau_0]$ such that
\begin{align}
\begin{split}
\big\| \bar{x}_{k|k-1} - x_{k|k-1} \big\| &\;\le\; \varepsilon_{\scn,x} + C_x\, \wct^{(m+1)/2}, \\
\big\| P_{k|k-1} - \Sigma_{k|k-1} \big\|_F &\;\le\; \varepsilon_{\scn,\Sigma} + C_\Sigma\, \wct^{(m+1)/2}.
\label{eq:thmdecomp}
\end{split}
\end{align}
where $\|\cdot\|$ denotes $L^2$ norm, and $\|\cdot\|_F$ denotes Frobenius norm.
\end{thm}

The proof of Theorem~\ref{thm:WSC-KF-error} can be found in Appendix~\ref{ap:main_pf}, and rests on a generalization (Lemma~\ref{lem:Lyons}) of \cite[Prp. 3.2]{lyons2004cubature} to handle certain polynomial-growth functions, enabling uniform Markov semigroup bounds under Assumption~\ref{as:smooth}.

The $O(\tau^{(m+1)/2})$ remainder is the local weak rate
on an interval of length $\tau$.  For $\tau\geq 1$ the
exponent $(m+1)/2$ no longer improves the bound as $m$
increases; one instead applies the same formula on each
piece of a partitioned grid, as in the recombination procedure of Section~\ref{sec:ws_cub}, in which case the error rate \eqref{eq:recom_err} is inherited.  
 
\begin{cor}[Comparison to CD-CKF]
\label{cor:cdckf-comparison}
Assume the hypotheses of Theorem~\ref{thm:WSC-KF-error}, and take the CD-CKF path map to be a first-order It\^o step (the state-dependent diffusion extension of \cite{arasaratnam2010cubature}). Then the path term of that CD-CKF is $O(\tau^2)$.  WSC-KF of degree $m=3$ has the same path order; of degree $m\ge 5$, the path term is $O(\tau^{(m+1)/2})$. The comparison is of local weak remainders for one predict step of length $\tau$.
\end{cor}

Both filters pay \emph{exactly} the same spatial-discretization price $\varepsilon_{\scn,x}, \varepsilon_{\scn,\Sigma}$, as neither the choice of $m$ nor the path-space cubature affects the spatial cubature rule. Corollary~\ref{cor:cdckf-comparison} is not a comparison with the additive It\^o--Taylor $1.5$ moment map of \cite{arasaratnam2010cubature}, which is closed form for \emph{additive} process noise and has local remainder $O(\tau^3)$. For \emph{state-dependent} process noise (our focus) the usual CD-CKF step is first-order It\^o; weak order-$2$ It\^o--Taylor schemes that avoid iterated integrals do so by replacing them with auxiliary random
variables: they do not match the mixed integrals in~\eqref{eq:sol_exp}, and a sampled implementation reintroduces Monte Carlo into the predict step. The Ninomiya--Victoir scheme \cite{ninomiya2008weak} is a degree-$5$ cubature and has the same global weak order as $m=5$; WSC-KF uses that cubature class as a deterministic, degree-$m$ product with the CD-CKF spatial rule.

\subsection{Discussion. Path Evaluation Complexity.} The guarantee of Theorem~\ref{thm:WSC-KF-error} demonstrates advantageous \emph{asymptotic} accuracy; this is the formal regime in which the accuracy of these prediction operations, such as truncated It\^o-Taylor expansions, are to be compared \cite{arasaratnam2010cubature}. However, in practice one will be interested in a non-negligible $\tau > 0$. Then, as discussed in Section~\ref{sec:recom}, one may tailor the path construction to obtain a desirable guarantee by subdividing $[0,\tau]$ into a finite grid. Furthermore, the advantage of this procedure necessitates $\scn \wcn$ path evaluations. We contrast with ensemble methods, such as the EnKF and PF, which incur a path evaluation along each of $n_\text{ens}$ particles, typically with $n_\text{ens} \gg \scn$. Thus, the complexity overhead of the WSC-KF can be isolated to the total number of cubature paths $\wcn$; for brevity, we refer to \cite{snow2025efficient} for in-depth discussion of the construction of these paths, their dependence on the Brownian dimension, and approaches for reducing this overhead through high-order recombination \cite{litterer2012high}.\footnote{It is also shown in \cite{snow2025efficient} that the RMSE of approximating SDE functional statistics by Wiener cubature evaluations can scale as $O(\wcn^{-1})$, in contrast to Monte Carlo scaling of $O(n_{mc}^{-1/2})$, with $n_{mc}$ the number of stochastic Monte Carlo paths. This is a complementary advantage of our method, which may be exploited in future filtering designs.}

\section{Numerical Examples.}
\label{sec:num}

We report experiments on the It\^o system
\begin{equation} \label{eq:test_system}
  \mathrm{d}x^{(1)}=\mathrm{d}W^{(1)},\qquad \mathrm{d}x^{(2)}=\sin(x^{(1)})\,\mathrm{d}W^{(2)},
\end{equation}
started at $(1, 0)$. The vector fields are clearly globally Lipschitz and $C^\infty$, and the hypotheses of Theorem~\ref{thm:WSC-KF-error} hold.  Figure~\ref{fig:numerics}(a--b) illustrate the weak rate of Theorem~\ref{thm:WSC-KF-error}, and Figure~\ref{fig:numerics}(c) score the same coordinate as a sequential nonlinear observation, against TME, the SE-UKF, the CD-CKF, the CD-EnKF, and a bootstrap particle filter. All runs use JAX on CPU, 64-bit arithmetic, and classical RK4.
WSC-KF uses degree-$3$ spherical cubature in the state
($\scn=2d_x=4$) and explicit Lyons--Victoir paths \cite{lyons2004cubature}.  Concatenation order $N_\tau$ does not raise the cubature degree $m$.

\begin{figure*}[t]
    \centering
    \includegraphics[width=0.98\linewidth]{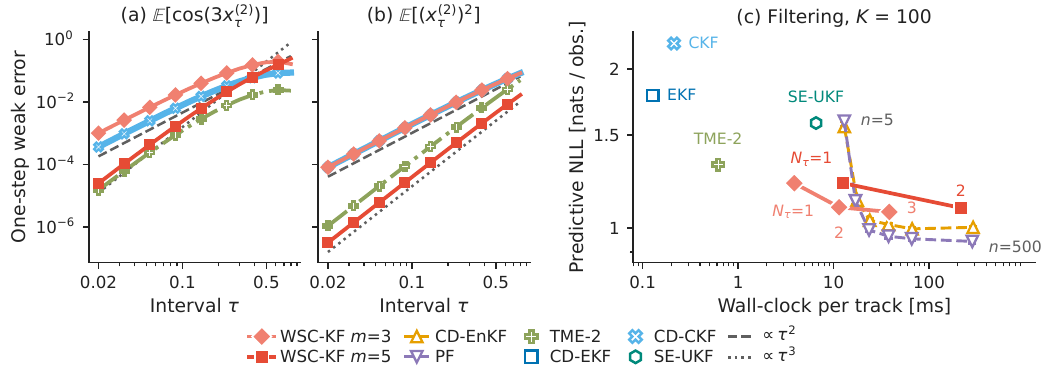}
    \caption{\textbf{Numerical Results.} On the left, one step weak error of the predict step with test functions \textbf{(a)} $\cos(3 x^{(2)})$ and \textbf{(b)} ${(x_\tau^{(2)})}^2$, both with Dirac law for $x_0$. On the right, \textbf{(c)} sequential mixture predictive NLL of noisy observations $y = \cos(3 x^{(2)}) + \varepsilon$.}
    \label{fig:numerics}
\end{figure*}

\subsection{Weak-Order Rates Confirmation}
\label{sec:levy-weak}

To isolate the Lyons--Victoir contribution from spatial cubature, we take a single predict step from a Dirac law $P_0=0$. We illustrate the weak error with respect to two test functions: $f(x) = \cos\bigl(3 x_\tau^{(2)}\bigr)$ and ${\bigl(x_\tau^{(2)}\bigr)}^2$. In either case, the test function clearly belongs to $C^\infty$. We estimate the ground truth value of the former via a Fourier expansion, and the latter is available in closed form.

Theorem~\ref{thm:WSC-KF-error} gives local weak error $O(\tau^{(m+1)/2})$ for WSC-KF, hence $O(\tau^2)$ at $m=3$ and $O(\tau^3)$ at $m=5$; Corollary~3 gives $O(\tau^2)$ for the CD-CKF. Figure~\ref{fig:numerics}(a--b) reports the weak error
against $\tau$ for $(m,N_\tau)=(3,1)$ and $(5,1)$, where recall $N_\tau$ is the subinterval count in the concatenated path construction in Section~\ref{sec:ws_cub}. We also show lines proportional to the predicted $\tau^2$ and $\tau^3$ error rates; in both cases, the CD-CKF and WSC-KF closely follow the predicted rate, especially at small intervals $\tau$.

\subsection{Demonstration of Path Law Matching}
\label{sec:levy}
The previous section discussed the weak error of each predict step; this plays a part in strong filtering performance, but is not itself sufficient. We next test filtering performance directly.

Figure~\ref{fig:numerics}(c) reports sequential mixture predictive negative log-likelihood (NLL) of the nonlinear observations
\begin{equation}
  y_k=\cos\bigl(3 x^{(2)}_{t_k}\bigr)+\varepsilon_k, \quad R^{1/2}=0.25,
\end{equation}
averaged over sixty independent observation trajectories of length $K = 100$.  Each method is scored on its predictive cloud, $-\log\sum_j w_j\mathcal N(y_k;h(x_j),R)$, rather than on a Kalman Gaussian of $y$ obtained after collapsing that cloud.  The observation map is a bounded nonlinear functional of $x^{(2)}$: matching the first two moments of $x$ or linearizing $h$ does not determine the law of $y$.

EKF, CKF, and the SE-UKF all suffer the most error in this example, all with $\mathrm{NLL} > 1.5$ nats. TME-2 provides better performance, and also lies along a Pareto frontier of compute cost and performance; WSC-KF provides a subsequent set of entries in the Pareto frontier, with lower error than TME-2 and faster computation than any EnKF or PF computations. WSC-KF slightly improves in moving from $m = 3$ to $m = 5$, and significantly improves in increasing $N_\tau$, both at some computational cost. Finally, the CD-EnKF and bootstrap PF each provide gold-standard performance with sufficient particle count, but are significantly more expensive than the (type-II) Gaussian filters.

\section{Conclusion}
We have developed a new nonlinear continuous-discrete filter, which inherits the stability and efficiency benefits of Gaussian approximation filters while enabling high-order path-law approximation in the filtering prediction step. We call this the Wiener-Space Cubature Kalman Filter (WSC-KF), since it relies on the theory of Wiener-space cubature to approximate this infinite-dimensional Wiener path-law using a system of deterministic ordinary differential equations. This unlocks not only higher-order approximation of model transition dynamics, but exploitation of efficient numerical schemes for parallel ODE solving and gradient computation. We have demonstrated the advantages of this approach both theoretically, by comparing its path-law approximation capability directly to that of the continuous-discrete cubature Kalman filter, and numerically against state-of-the-art baselines. 
\section*{AI Use Declaration}

Generative AI was used to assist in code development, to provide feedback on the manuscript, and check the proofs. All code was reviewed and checked for correctness by the authors. The research idea, mathematical formulation and analysis, experimental design, and all writing are solely the work of the authors.

\bibliographystyle{plain}        
\bibliography{Bibliography}           

\begin{thebibliography}{10}

\bibitem{arasaratnam2010cubature}
Ienkaran Arasaratnam, Simon Haykin, and Thomas~R Hurd.
\newblock Cubature {K}alman filtering for continuous-discrete systems: Theory and simulations.
\newblock {\em IEEE Transactions on Signal Processing}, 58(10):4977--4993, 2010.

\bibitem{beard1999nonlinear}
Randal Beard, John Kenney, Jacob Gunther, Jonathan Lawton, and Wynn Stirling.
\newblock Nonlinear projection filter based on {G}alerkin approximation.
\newblock {\em Journal of Guidance, Control, and Dynamics}, 22(2):258--266, 1999.

\bibitem{berzuini1997dynamic}
Carlo Berzuini, Nicola~G Best, Walter~R Gilks, and Cristiana Larizza.
\newblock Dynamic conditional independence models and {M}arkov chain {M}onte {C}arlo methods.
\newblock {\em Journal of the American Statistical Association}, 92(440):1403--1412, 1997.

\bibitem{crisan2013kusuoka}
Dan Crisan and Salvador Ortiz-Latorre.
\newblock A {K}usuoka--{L}yons--{V}ictoir particle filter.
\newblock {\em Proc. R. Soc. Lond. Ser. A Math. Phys. Eng. Sci}, 469(20130076), 2013.

\bibitem{daum2006meshfree}
Fred Daum and Mikhail Krichman.
\newblock Meshfree adjoint methods for nonlinear filtering.
\newblock In {\em 2006 IEEE Aerospace Conference}. IEEE, 2006.

\bibitem{doucet2001introduction}
Arnaud Doucet, Nando De~Freitas, and Neil Gordon.
\newblock An introduction to sequential {M}onte {C}arlo methods.
\newblock In {\em Sequential {M}onte {C}arlo Methods in Practice}, pages 3--14. Springer, 2001.

\bibitem{friz2014course}
Peter~K Friz, Martin Hairer, et~al.
\newblock {\em A course on rough paths: with an introduction to regularity structures}, volume~14.
\newblock Springer, 2014.

\bibitem{fuller1969analysis}
AT~Fuller.
\newblock Analysis of nonlinear stochastic systems by means of the {F}okker--{P}lanck equation.
\newblock {\em International Journal of Control}, 9(6):603--655, 1969.

\bibitem{jazwinski2007stochastic}
Andrew~H Jazwinski.
\newblock {\em Stochastic Processes and Filtering Theory}.
\newblock Courier Corporation, 2007.

\bibitem{jia2013high}
Bin Jia, Ming Xin, and Yang Cheng.
\newblock High-degree cubature kalman filter.
\newblock {\em Automatica}, 49(2):510--518, 2013.

\bibitem{kalman1960new}
R.~E. Kalman.
\newblock A new approach to linear filtering and prediction problems.
\newblock {\em Journal of Basic Engineering}, 82(1):35--45, 03 1960.

\bibitem{karatzas2014brownian}
Ioannis Karatzas and Steven Shreve.
\newblock {\em Brownian motion and stochastic calculus}.
\newblock springer, 2014.

\bibitem{kloeden1992numerical}
Peter~E. Kloeden and Eckhard Platen.
\newblock {\em Numerical Solution of Stochastic Differential Equations}.
\newblock Springer-Verlag, Berlin, 1992.

\bibitem{lee2013adaptive}
Wonjung Lee and Terry Lyons.
\newblock The adaptive patched particle filter and its implementation.
\newblock {\em arXiv preprint arXiv:1311.6755}, 2013.

\bibitem{litterer2012high}
Christian Litterer and Terry Lyons.
\newblock High order recombination and an application to cubature on {W}iener space.
\newblock {\em The Annals of Applied Probability}, pages 1301--1327, 2012.

\bibitem{lyons2014series}
Simon~MJ Lyons, Simo S{\"a}rkk{\"a}, and Amos~J Storkey.
\newblock Series expansion approximations of brownian motion for non-linear kalman filtering of diffusion processes.
\newblock {\em IEEE Transactions on Signal Processing}, 62(6):1514--1524, 2014.

\bibitem{lyons2004cubature}
Terry Lyons and Nicolas Victoir.
\newblock Cubature on {W}iener space.
\newblock {\em Proceedings of the Royal Society of London. Series A: Mathematical, Physical and Engineering Sciences}, 460(2041):169--198, 2004.

\bibitem{masaya2021ensemble}
Masaya Murata, Isao Kawano, and Koichi Inoue.
\newblock Ensemble kalman filter for continuous-discrete state-space models.
\newblock In {\em 2021 60th IEEE Conference on Decision and Control (CDC)}, pages 6608--6613, 2021.

\bibitem{ninomiya2008weak}
Syoiti Ninomiya and Nicolas Victoir.
\newblock Weak approximation of stochastic differential equations and application to derivative pricing.
\newblock {\em Applied Mathematical Finance}, 15(2):107--121, 2008.

\bibitem{raviart2006analysis}
Pierre-Arnaud Raviart.
\newblock An analysis of particle methods.
\newblock In {\em Numerical Methods in Fluid Dynamics: Lectures given at the 3rd 1983 Session of the Centro Internationale Matematico Estivo (CIME) held at Como, Italy, July 7--15, 1983}, pages 243--324. Springer, 2006.

\bibitem{ribeiro2004kalman}
Maria~Isabel Ribeiro.
\newblock {K}alman and extended {K}alman filters: Concept, derivation and properties.
\newblock {\em Institute for Systems and Robotics}, 2004.

\bibitem{litterer2026high}
Emilio Rossi~Ferrucci, Timothy Herschell, Christian Litterer, and Terry Lyons.
\newblock High-degree cubature on {W}iener space through unshuffle expansions.
\newblock {\em Proceedings of the Royal Society of London. Series A, Mathematical and Physical Sciences}, 2026.

\bibitem{sarkka2019applied}
Simo S{\"a}rkk{\"a} and Arno Solin.
\newblock {\em Applied Stochastic Differential Equations}.
\newblock Cambridge University Press, 2019.

\bibitem{snow2025efficient}
Luke Snow and Vikram Krishnamurthy.
\newblock Efficient neural {SDE} training using {W}iener-space cubature.
\newblock {\em arXiv preprint arXiv:2502.12395}, 2025.

\bibitem{stroock2007multidimensional}
Daniel~W Stroock and SR~Srinivasa Varadhan.
\newblock {\em Multidimensional Diffusion Processes}.
\newblock Springer, 2007.

\bibitem{wan2000unscented}
Eric~A Wan and Rudolph Van Der~Merwe.
\newblock The unscented {K}alman filter for nonlinear estimation.
\newblock In {\em Proceedings of the IEEE 2000 Adaptive Systems for Signal Processing, Communications, and Control Symposium}, pages 153--158. IEEE, 2000.

\bibitem{zhao2020taylor}
Zheng Zhao, Toni Karvonen, Roland Hostettler, and Simo S{\"a}rkk{\"a}.
\newblock Taylor moment expansion for continuous-discrete gaussian filtering.
\newblock {\em IEEE Transactions on Automatic Control}, 66(9):4460--4467, 2020.

\end{thebibliography}

\appendix

\section{Supporting Results}
\label{ap:supp}

Here we provide several results which support the proof of Theorem~\ref{thm:WSC-KF-error}. Lemma~\ref{lem:usb} provides a uniform Markov semigroup bound, and follows from standard results in SDE theory. Lemma~\ref{lem:Lyons} extends Proposition 3.2 of \cite{lyons2004cubature} from \emph{bounded} functions to two particular polynomial-growth functions we will need for bounding the Markov semigroup moments. 
\begin{lem}[Uniform Semigroup Boundedness]
\label{lem:usb}
Under Assumption~\ref{as:smooth}, for any finite set of points $\mathcal{S} = \{\zeta_i\}_{i=1}^{|\mathcal{S}|}$ and $\wct_0>0$, there exists $M_{\mathcal{S}}> 0$ such that $\sup_{\zeta\in\mathcal{S},\wct<\wct_0}\|\E[x_{\wct} |x_0=\zeta]\| \leq M_{\mathcal{S}}$.
\end{lem}
\vspace{-0.4cm}
\begin{pf}
This follows directly from Theorem 2.9 of \cite{karatzas2014brownian}, by simply taking the constant $M_{\mathcal{S}}$ to be the maximum r.h.s. of \cite{karatzas2014brownian} equation 2.15 over the finite set $\mathcal{S}$.\qed
\end{pf}

\begin{lem}
\label{lem:Lyons}
Let $\chi(x)=x_p$ or $\chi(x)=x_px_q$.  Under Assumption~\ref{as:smooth}, for every
finite set $S=\{\zeta_i\}$ and every $\tau_0>0$ there is
$C_{S,\tau_0}<\infty$ such that, for all $\zeta\in S$ and
$t\in(0,\tau_0]$,
\[\bigl|\mathbb E[\chi(x_t)\mid x_0=\zeta] -\textstyle\sum_{j=1}^{N_c}\lambda_j\chi(\Phi_{t,\zeta}(w_j))\bigr| \le C_{S,\tau_0}\,t^{(m+1)/2}. \]
where $\Phi_{t,\zeta}(w_j)$ denotes the solution \eqref{eq:Phi_dyn}.
\end{lem}
\vspace{-0.5cm}
\begin{pf}
The remainder in \cite[Prop.~3.2]{lyons2004cubature} is controlled by
$V_{i_1}\cdots V_{i_k}\chi$ along the SDE and along each cubature ODE, for $(i_1,\ldots,i_k)\in A_{m+2}\setminus A_m$.  Assumption~\ref{as:smooth} gives $|V_{i_1}\cdots V_{i_k}\chi(x)|\le C(1+\|x\|^r)$.  It is therefore enough that these trajectories have uniformly finite moments of order $r$ on $[0,\tau_0]$.
The fields are globally Lipschitz, hence of linear growth.  For the SDE, \cite[Thm.~2.9]{karatzas2014brownian} yields
$\mathbb E[\sup_{s\le\tau_0}\|x_s\|^r\mid x_0=\zeta]<\infty$ uniformly for $\zeta\in S$. For each cubature ODE~(7), $w_j$ has finite $1$-variation on $[0,\tau_0]$. Gronwall's inequality therefore gives
\begin{align*}
&\sup_{s\le\tau_0}\|\Phi_{s,\zeta}(w_j)\| \\&\,\,\le
\bigl(\|\zeta\|+C(\tau_0+|w_j|_{1\text{-var}})\bigr)
\exp\bigl(C(\tau_0+|w_j|_{1\text{-var}})\bigr),
\end{align*}
uniformly in the finite set of paths $j$ and nodes $\zeta\in S$. The remainder constant is finite, and \cite[Prop.~3.2]{lyons2004cubature} yields the
rate.  \qed
\end{pf}

\vspace{-0.4cm}
\section{Main Result Proof}
\label{ap:main_pf}
\vspace{-0.5cm}
\begin{pf*}{PROOF of Theorem~\ref{thm:WSC-KF-error}}

\textit{Mean Bound}: We denote the set of spatial cubature nodes by $\{\zeta_i\}_{i=1}^{\scn}$, with weights $\{c_i\}_{i=1}^{\scn}$. Also denote $\mu(x) := \E[x_k  |x_{k-1} = x],\,\,\, \Sigma(x) := \text{Cov}(x_k | x_{k-1} = x)$
For each spatial node, let us write
$\hat{\mu}(\zeta_i) := \sum_{j=1}^{\wcn}\lambda_j\Phi_{\Dt,\zeta_i}(w_j)$
as the Wiener cubature forecast, approximating the statistic $\E[x_{t+\Dt} |x_t = \zeta_i]$. Then, $\bar{x}_{k|k-1}$ in \eqref{eq:fore_stat} is simply $\bar{x}_{k|k-1}  = \sum_{i=1}^{\scn}c_i\hat{\mu}(\zeta_i)$
Let us also write 
$\tilde{x}_{k|k-1} := \sum_{i=1}^{\scn}c_i\mu(\zeta_i)$
as the exact (w.r.t. true forecast law) prediction induced by the approximate spatial cubature. Then, we may split the total WSC-KF forecast mean error by triangle inequality:
$\|\bar{x}_{k|k-1} - x_{k|k-1}\| \leq \|\bar{x}_{k|k-1} - \tilde{x}_{k|k-1}\| + \|\tilde{x}_{k|k-1} - x_{k|k-1}\|$
where the first term on the r.h.s. quantifies the error w.r.t. only the Wiener-cubature approximation of the transition law, and the second term quantifies the error introduced by the spatial cubature approximation. Now the first term can be bounded straightforwardly as $\|\bar{x}_{k|k-1} - \tilde{x}_{k|k-1}\| = \left\|\sum_{i=1}^{\scn} c_i (\hat{\mu}(\zeta_i) - \mu(\zeta_i))\right\| \leq \sum_{i=1}^{\scn}c_i\|\hat{\mu}(\zeta_i) - \mu(\zeta_i)\| = O(\Dt^{(m+1)/2})$
where the last equality follows from Lemma~\ref{lem:Lyons}. The second term corresponds exactly to the error incurred by the spatial cubature approximation.

\textit{Covariance Bound}:
We begin by writing the degree-2 $C^{\infty}$-smooth polynomial function $f_{pq}(\xi) = \xi_p\xi_q$. Then, by Lemma~\ref{lem:Lyons} we obtain the entrywise bound, with $x_\wct$ the solution of \eqref{eq:strat_sde} at time $\wct$ and any initial point $\zeta_i \in S$, $\left|\E[(x_\wct)_p(x_\wct)_q - \sum_{j=1}^{\wcn} \lambda_j [\Phi_{\wct,\zeta_i}(w_j)]_p[\Phi_{\wct,\zeta_i}(w_j)]_q\right|\quad = O(\wct^{(m+1)/2}), \quad \forall p,q\in[\xdim], x_0\in\reals^{\xdim}$
and thus $ \|\E[x_\wct x_\wct^{\top}] - \sum_{j=1}^{\wcn} \lambda_j \Phi_{\wct,\zeta_i}(w_j)\Phi_{\wct,\zeta_i}(w_j)^{\top}\|_F = O(\wct^{(m+1)/2})$,
where $\|\cdot\|_F$ denotes Frobenius norm. Denote $P(\zeta_i)$ the analogue of covariance $\Sigma(\zeta_i)$ under the WSC-KF predict step. Now, for each $\zeta_i$, we may decompose $P(\zeta_i) - \Sigma(\zeta_i) = \left[\sum_{j=1}^{\wcn} \lambda_j \Phi_{\wct,\zeta_i}\Phi_{\wct,\zeta_i}^{\top} -\E[x_\wct x_{\wct}^{\top}]\right]  - \left[\hat{\mu}(\zeta_i)\hat{\mu}(\zeta_i)^\top - \mu(\zeta_i) \mu(\zeta_i)^\top \right]$. 
The first term in the r.h.s. of this equality was bounded in Frobenius norm in the previous line. 

Let us for now simplify notation by absorbing $\zeta_i$. The second term can be bounded using the identity $\hat{\mu}\hat{\mu}^\top - \mu\mu^\top = (\hat{\mu} - \mu)\mu^\top + \hat{\mu}(\hat{\mu} - \mu)^\top$, and splitting as $\| (\hat{\mu} - \mu)\mu^\top + \hat{\mu}(\hat{\mu} - \mu)^\top\|_F \leq \|(\hat{\mu} - \mu)\mu^\top\|_F + \|\hat{\mu}(\hat{\mu} - \mu)^\top\|_F\leq \|\hat{\mu} - \mu\|\|\mu\| + \|\hat{\mu}\|\|\hat{\mu} - \mu\|$
where the inequality is simply the triangle inequality, and the equality is a property of the Frobenius norm. 
Then, using Lemma~\ref{lem:Lyons} again the difference terms are controlled by the $O(\wct^{(m+1)/2})$ rate. By Lemma~\ref{lem:usb}, we have $\|\mu\| \leq M_{\mathcal{S}}$, and also $\|\hat{\mu}\| \leq \|\mu\| + \|\hat{\mu} - \mu\| \leq M_{\mathcal{S},\wct} + O(\wct^{(m+1)/2})$. Thus, combining gives $\|\hat{\mu}(\zeta_i)\hat{\mu}(\zeta_i)^\top - \mu(\zeta_i)\mu(\zeta_i)^\top\|_F = O(\wct^{(m+1)/2})$, and then the triangle inequality produces
\begin{align}
\label{eq:PSig_diff}
    \|P(\zeta_i) - \Sigma(\zeta_i)\|_F = O(\wct^{(m+1)/2})
\end{align}
Now, average over the $\scn$ nodes with weights $c_i$. By the law of total variance, $P_{k|k-1} = \sum_{i=1}^{\scn}  c_i P(\zeta_i) + \sum_{i=1}^{\scn}c_i\big(\hat\mu(\zeta_i)-\bar x_{k|k-1}\big)^{\otimes 2},
    \tilde\Sigma_{k|k-1} := \sum_{i=1}^{\scn} c_i\Sigma(\zeta_i) + \sum_{i=1}^{\scn}c_i\big(\mu(\zeta_i)-\tilde x_{k|k-1}\big)^{\otimes 2}$,
where $v^{\otimes 2}:=vv^\top$, $\tilde x_{k|k-1}=\sum_{i=1}^{\scn} c_i\mu(\zeta_i)$ as above, and $\tilde\Sigma_{k|k-1}$ is the spatial-cubature-only target underlying $\varepsilon_{\scn,\Sigma}$; that is, $\|P_{k|k-1} - \Sigma_{k|k-1}\|_F \leq \|P_{k|k-1} - \tilde{\Sigma}_{k|k-1}\|_F + \|\tilde{\Sigma}_{k|k-1} - \Sigma_{k|k-1}\|_F$, where $\|\tilde{\Sigma}_{k|k-1} - \Sigma_{k|k-1}\|_F = \varepsilon_{\scn,\Sigma}$. To bound $\|P_{k|k-1} - \tilde{\Sigma}_{k|k-1}\|_F$, we  subtract and apply $a^{\otimes2}-b^{\otimes2}=(a-b)a^\top+b(a-b)^\top$ termwise, with $a_i:=\hat\mu(\zeta_i)-\bar x_{k|k-1}$ and $b_i:=\mu(\zeta_i)-\tilde x_{k|k-1}$: the first term is $O(\wct^{(m+1)/2})$ by \eqref{eq:PSig_diff}; the second (between-node) term is likewise $O(\wct^{(m+1)/2})$, since
$a_i-b_i = \big(\hat\mu(\zeta_i)-\mu(\zeta_i)\big) - \big(\bar x_{k|k-1}-\tilde x_{k|k-1}\big) = O(\wct^{(m+1)/2})$
by Lemma~\ref{lem:Lyons} and the mean bound above, while $\|a_i\|,\|b_i\|$ are $\wct$-uniformly bounded by Lemma~\ref{lem:usb}. Hence
\[
\|P_{k|k-1}-\tilde\Sigma_{k|k-1}\|_F \leq C_\Sigma\,\wct^{(m+1)/2}.
\]
\qed
\end{pf*}

\end{document}